\documentclass[conference,10pt]{IEEEtran} 
\IEEEoverridecommandlockouts
\usepackage{amsmath,amssymb,amsfonts}
\usepackage[ruled]{algorithm}
\usepackage[noend]{algpseudocode}
\usepackage{multirow}
\usepackage{graphicx}
\usepackage{textcomp}
\usepackage{xcolor}
\usepackage{comment}
\usepackage{biblatex}
\algblockdefx[Upon]{Upon}{EndUpon}%
[1]{{\bf upon} (#1) do}%
{{\bf end upon}}

\newcommand{\mb}[1]{\textcolor{red}{#1}}

\newtheorem{theorem}{Theorem}
\newtheorem{defi}{Definition}

\newenvironment{proof}{\noindent {\sc Proof. }}{\hfill$\Box$}

\newcounter{cur_exp}[section]
\newcounter{old_exp}[section]
\newcommand{\curexp}{\arabic{cur_exp}}
\newcommand{\oldexp}{\arabic{old_exp}}
\newcommand{\gameold}{\ensuremath{G_{\oldexp}}}
\newcommand{\gamecur}{\ensuremath{G_{\curexp}}}

\newcommand{\initexp}{\vspace{5pt}%
                      \setcounter{cur_exp}{0} %
                      \setcounter{old_exp}{-1} %
                      \noindent %
                      \textbf{Game} \ensuremath{\mathbf{\curexp}}. %
                     }

\newcommand{\newexp}{\vspace{4pt}%
                     \addtocounter{cur_exp}{1}%
                     \addtocounter{old_exp}{1}%
                     \noindent
                     \textbf{Game} \ensuremath{\mathbf{\curexp}}. %
                     }

\def\Adv{{\sf Adv}}
\def\cA{{\mathcal A}}
\def\cO{{\mathcal O}}
\def\cS{{\mathcal S}}

\def\cD{{\mathcal D}}
\begin{document}

\title{Chirotonia: A Scalable and Secure e-Voting Framework based on Blockchains and Linkable Ring Signatures}
\author{\IEEEauthorblockN{Antonio Russo \\
Antonio Fernández Anta}
\IEEEauthorblockA{\textit{IMDEA Networks Institute} \\
Madrid, Spain \\
email: antonio.russo@imdea.org \\
email: antonio.fernandez@imdea.org}
\and
\IEEEauthorblockN{María Isabel González Vasco}
\IEEEauthorblockA{\textit{Departamento de Matemática Aplicada} \\
\textit{Universidad Rey Juan Carlos}\\
Mostoles, Madrid, Spain \\
email: mariaisabel.vasco@urjc.es}
\and
\IEEEauthorblockN{Simon Pietro Romano}
\IEEEauthorblockA{\textit{DIETI} \\
\textit{University of Napoli, Federico II}\\
Naples, Italy \\
email: spromano@unina.it}
}

\maketitle

\begin{abstract}
In this paper we propose a comprehensive and scalable framework to build secure-by-design e-voting systems. Decentralization, transparency, determinism, and untamperability of votes are granted by dedicated smart contracts on a blockchain, while voter authenticity and anonymity are achieved through (provable secure) linkable ring signatures. These, in combination with suitable smart contract constraints, also grant protection from double voting. Our design is presented in detail, focusing on its security guarantees and the design choices that allow it to scale to a large number of voters. Finally, we present a proof-of-concept implementation of the proposed framework, made available as open source.
\end{abstract}


\section{Introduction}

Building an e-voting system represents a big challenge, due to strong requirements in terms of confidentiality, integrity and availability. Such a system must preserve these properties not only when facing external attackers, but also in the presence of malicious insiders, like regular or administrative users who misbehave and try to undermine the correct operation of the e-voting system.
It is usually expected from an electronic voting scheme that it provides at least the same guarantees as those of a classic (paper-based) voting system, while adding the advantages of a modern electronic system, like remote participation and information processing speed. 

In this paper we propose a decentralized, secure, and scalable framework for e-voting systems based on linkable ring signatures, a blockchain, and smart contracts. The blockchain provides a reliable repository for the data used by the e-voting system, and the support for the smart contracts that define the computation performed in each phase of a voting process. The blockchain also relieves from the need of having a centralized solution for a verified bulletin board. Then, the voting ballots are signed with linkable ring signatures, which is a cryptographic tool that allows only legal voters to sign, and detects duplicated signers without revealing identities or ballot contents. In this paper we also adopt a linkable ring signature scheme using Elliptic Curve Cryptography (ECC), which grants very short keys and signatures, hence making the framework scalable.

\subsection{Related Work}

There are a number of secure e-voting systems proposed in the literature. The Helios system~\cite{DBLP:conf/uss/Adida08} is generally used as a reference due to its simple voting workflow, and its privacy and verifiability properties achieved through zero-knowledge proofs~\cite{DBLP:journals/siamcomp/GoldwasserMR89}. Its main drawback is that it requires the existence of a trusted central server, which is responsible for identifying voters, as well as for  storing, shuffling, decrypting, and publishing  ballots.

The Belenios system~\cite{DBLP:conf/birthday/CortierGG19} is based on Helios but improves it by providing universal verifiability, even in the presence of a dishonest bulletin board (the element used for making data publicly available for verifiability
checks). As Helios, Belenios has a complete open source implementation that is able to handle a reasonably large number of voters.
However, also as Helios, it still requires a trusted central server for the major process phases in order to ensure vote correctness.

Benaloh and Tuinstra~\cite{DBLP:conf/stoc/BenalohT94} proposed a receipt-free voting system able to guarantee coercion resistance, ballot privacy, and verifiability via interactive zero-knowledge proofs. Unfortunately, their system requires a trusted central authority that continuously interacts with the voters during the voting process. Moreover, voters must not be able to communicate with external entities during
the voting process. Compared with classic paper ballot voting, the overhead of this system in terms of deployment effort and actions required from the voters is unfortunately not trivial. The advantages provided are the determinism and quick aggregations of the result. Remote voting is hardly achievable.

The D-DEMOS suite \cite{DBLP:journals/compsec/ChondrosZZDMPDK19} (which is built on top of Helios \cite{DBLP:conf/uss/Adida08}) takes a step further in the guarantees and properties that are ensured. It adopts Verified Secret Sharing \cite{DBLP:conf/eurocrypt/Pedersen91a}, Zero Knowledge Proofs, and applies segregation of duties for the main tasks of the voting process. All the properties are provided in the assumption that, for each vote, an Election Authority setups the system, distributes the required cryptography information to all the involved actors (included voters, with different data for each of them), and gets destroyed without colluding with any party. A key requirement of D-DEMOS in order to ensuring the system properties is that the Election Authority gets destroyed and does not share the secret information it generated on behalf of the other entities. In fact, even if its computations can be verified, it is not possible to avoid that a malicious Election Authority shares secret information of voters with other parties; this would enable them to vote on behalf of the voters. In our approach all the tasks necessary for the voting session setup are publicly verifiable, they do not involve any secret generation on behalf of other parties and, for each new vote session, voters can directly compute and submit a new valid ballot with the information they already hold.

To our knowledge, in the distributed and decentralized setting we want to provide, there are two main relevant works. First, Liu and Wang~\cite{DBLP:journals/iacr/LiuW17} propose an e-voting system that uses a blockchain paired with blind signatures.
Voters' privacy is granted by computing blind signatures of ballots, with third parties being in charge of identifying voters and signing back data. This requires the identification process to be run for each vote and for each identification authority. Then the blockchain acts as an auditable communication channel among parties. In contrast, in our proposal the identification process can be reused over many voting sessions, voters can submit their ballots without waiting for other entities to come into play, and the blockchain is an active validator for the used data.

Second, Lyu et al.~\cite{DBLP:conf/trustcom/LyuJWNAF19} is the existing work more relevant to our framework, since it also combines a blockchain, smart contracts, and ring signatures. Their voting protocol is very similar to the one proposed here, yet there are some relevant differences. For instance, in their proposal, voters are responsible for running a distributed key generation protocol, thus involving a significant amount of processing and communication for large pools of voters, making it hard to go beyond a small number of users (the authors use up to $40$ users in their tests). In contrast, as we will see, our design smoothly scales up to thousands of voters.

\subsection{Contributions}

In this work we propose a framework to build e-voting systems based on a blockchain, smart contracts, and linkable ring signatures. The framework is realized as a set of protocols and tools that can be used in a wide range of e-voting scenarios. More precisely, we provide:

\begin{itemize}
    \item A modular approach to build distributed, secure, verifiable, and scalable e-voting systems.
    \item A set of tools to implement such systems in different e-voting scenarios.
    \item A formal security analysis supporting our choice of a linkable ring signature scheme, as proposed in~\cite{DBLP:conf/acisp/LiuWW04}.
    \item A proof-of-concept implementation that can be run in any Ethereum-compliant blockchain, demonstrating the practical feasibility of our approach.
\end{itemize}

\subsection{Paper Roadmap}
In Section~\ref{sec:preliminaries} we introduce the blockchain properties we leverage and our main cryptographic tool: linkable ring signatures. Section~\ref{sec:lsag-schema} presents the concrete signature scheme adopted, describing the security properties it satisfies (our full paper version provides a complete formal security analysis), while Section~\ref{sec:storageEfficiency} focuses on its storage efficiency. Our e-voting framework is presented in Section~\ref{sec:voting-framework}, where we describe the involved actors and their interactions. Section~\ref{sec:deployment} provides the details of the \textit{identification} process and the proposed ways to guarantee \textit{confidentiality} and \textit{anonymity}.  Next, a complete list of  features achieved by our design is depicted in Section~\ref{sec:requirements}.
Finally, Section~\ref{sec:proof-of-concept} presents our choices for the implemented proof of concept, and Section~\ref{sec:conclusion} wraps up with some final conclusions.

\section{Building Blocks}
\label{sec:preliminaries}

\subsection{Blockchain and Smart Contracts}
The chain-of-blocks data structure, paired with a Byzantine Fault Tolerant consensus  (see~\cite{DBLP:conf/osdi/CastroL99,DBLP:journals/toplas/LamportSP82}) results in a powerful and nifty way to realise an untamperable Distributed Ledger. We will not go into details of blockchain foundations and its wide range of implementations. Indeed, we model the blockchain properties as a replicated database and a deterministic execution environment, resilient to a bounded set of Byzantine failures.
We refer to the code that is stored and executed in the blockchain as Smart Contract~\cite{DBLP:journals/firstmonday/Szabo97} and to the activation of a piece of it as a transaction. Thanks to the chain-of-blocks data structure, any modification to the system state is audited in an untamperable manner on the blockchain itself as part of the state (e.g., a transaction that calls a function in a smart contract is stored including information on the called function and all the parameters passed as arguments).

\subsection{Linkable Ring Signatures}
\label{sec:RingSignature}

Ring signatures are cryptographic constructions, first introduced by  Rivest, Shamir and Tauman~\cite{DBLP:conf/asiacrypt/RivestST01}, allowing for the creation of signatures on behalf of a group of signers, while hiding the identity of the actual signer among the group.
In~\cite{DBLP:conf/acisp/LiuWW04}, so-called \emph{linkable ring signatures} were introduced as an extension to ring signatures, adding the feature of providing a public way of determining whether two signatures have actually been produced by the same signer. As pointed out in~\cite{DBLP:conf/eurocrypt/0001DHKS19}, these constructions are particularly suited for e-voting schemes, as linkability prevents voters from casting multiple votes while anyonymity is preserved within the group of registered voters. Formally:

\begin{defi} [Linkable ring signature scheme] A \emph{linkable ring signature scheme} is defined by four probabilistic polynomial-time algorithms $(\textsf{KeyGen, Sign, Verify, Link})$ such that:
\begin{itemize}
\item{$\textsf{KeyGen}(1^{\ell})$}, the key generation algorithm, takes as an input the security parameter $\ell$ and outputs a pair $(pk,sk)$ of public (verification) and secret (signing) keys,
\item{$\textsf{Sign}(1^{\ell}, sk_\pi, pk_1,\dots, pk_n,m)$},  the signature algorithm,  takes as input the security parameter $\ell,$ a secret key $sk_\pi,$ for some $\pi\in\{1,\dots, n\}$, a message $m$ (which is assumed to belong to a public message space $\mathcal{M}_{\ell}$) and a list of public keys $pk_1,\dots, pk_n$, and outputs a signature $\sigma.$ 
\item{$\textsf{Verify}(1^{\ell}, pk_1,\dots, pk_n,m,\sigma)$}, the verification algorithm, takes as input the security parameter $\ell,$ a list of public keys $pk_1,\dots, pk_n$, a message  $m\in \mathcal{M}_{\ell}$ and a signature $\sigma$, and outputs a bit $b\in \{0,1\}$ (namely, $1$ if the signature is recognized as valid w.r.t. the list of public keys, and $0$ otherwise).
\item{$\textsf{Link}(\sigma_1,\sigma_2,m_1,m_2)$}, the linking algorithm, takes as an input two signatures, $\sigma_1,\sigma_2$, and two messages, $m_1,m_2$, and outputs a bit $\beta\in \{0,1\},$ indicating whether the two signatures have been produced by the same user or not.  
\end{itemize}
\end{defi}

There are two obvious correctness requirements to add to the above definition imposing that, indeed, properly constructed signatures can be verified and linked as intended. Furthermore, our usage of a linkable signature scheme requires that certain security properties are fulfilled, which we informally list below:
\begin{itemize}
\item{\it Unforgeability}:  only those belonging to the designed group of voters can produce a valid signature (which will be validated with respect to the election census).

\item{\it Anonymity:} the actual identity of a signer remains hidden within the signing group.

\item{\it Linkability:} different uses of the same secret key can be identified through the public information.  

\item{\it Non-framability:}  it is not possible for an adversary to interfere with the correctness of linkability, even if he/she is able to corrupt (a polynomial number of) group members. Namely, even seeing a number of signatures of a prescribed honest user, and controlling some legitimate signers, he/she should not be able to produce a new signature that will link with the honest user without his/her secret key
\footnote{Deviating from~\cite{DBLP:conf/acisp/LiuWW04} we  work with the (stronger, and more realistic) definition from~\cite{DBLP:conf/eurocrypt/0001DHKS19} where the case of colluding malicious users is captured, as the adversary is allowed to query for (a polynomial number of) secret keys.}

\end{itemize}

In the appendices, we give formal definitions (Appendix~\ref{appendix:ec-lsag}) and security proofs (Appendix~\ref{appendix:security-proofs}) for the concrete instantiation of a linkable ring signature scheme we use in our design, which we describe in the next section and have implemented in our proof of concept (see Section~\ref{sec:proof-of-concept}).  While we follow the approach of~\cite{DBLP:conf/acisp/LiuWW04}, it is worth mentioning two main differences in our security analysis. First, we allow for corruptions in the unforgeability definition, which  are not considered in~\cite{DBLP:conf/acisp/LiuWW04}. This is particularly relevant for our use case, as a collusion of malicious users of the system should indeed not be able to produce a valid signature that can be verified using a set of public keys from honest users. Second, we include non-framability as a desirable property, which in particular for our use case implies that a collusion of voters may not produce a valid signature (e.g., vote) that can be linked with another (honest) user, thus preventing him/her from casting a vote.

\section{Proposed Linkable Ring Signature Scheme}

\subsection{Signature Scheme}
\label{sec:lsag-schema}
Elliptic curve cryptography (ECC) constructions make use of an elliptic curve defined over a finite field $K,$ which is fixed and made public through two coefficients $a,b$ (see~\cite{ECC}). Namely, the curve is defined as the set of points $(x,y)\in K^2$  fulfilling $y^2 = x^3 +ax+b.$ Moreover,  a group law $+$  is defined on this set of points, and typically a cyclic subgroup $E$ of this group is fixed,  selecting a generator $G.$ Under certain conditions, it is assumed that the associated \emph{discrete logarithm problem} --- {which in additive notation is described as the problem of finding the integer $a$ such that a given point $P\in E$ is the result of adding up $G$ with itself $a$ times, i.e., $P =a\cdot G$,}  --- is computationally hard. Moreover, many cryptographic constructions rely on the hardness of  two related problems:  \emph{Computational} and \emph{Decisional Diffie Hellman}, usually referred to as CDH and DDH (see~\cite{DH} for details on these problems).

\medskip

We describe now our modification of the proposal in~\cite{DBLP:conf/acisp/LiuWW04}. Once the security parameter $\ell$ is fixed, we assume there is a public  directory specifying a cyclic group $E$ of order $g$ within the group of points of an elliptic curve over a finite field $\mathbb{F}_p$ of  size $p$ ($g, p$ are polynomial in $\ell$). Moreover, there is also a fixed generator $G$ for $E$ that we assume publicly known. Furthermore, two different (public) hash functions $H$ and $\mathit{H2P}$ are chosen and made public:
$$
        H:  \{0,1\}^* \rightarrow \mathbb{F}_p, \,  
        \mathit{H2P}:  \{0,1\}^* \rightarrow E.
$$

 With these ingredients, a linkable ring signature scheme can be  defined as follows:
\medskip 

\noindent{\bf Key Generation.} $\textsf{KeyGen}(1^{\ell})$ takes as an input the security parameter $\ell$ and selects uniformly at random  an integer $sk$ within $\{1,\dots, g-1\}.$ Then, it sets $pk = sk \cdot G$ and outputs the resulting pair $(pk,sk)$ of public (verification) and secret (signing) keys.

\medskip

\noindent{\bf Signature Generation.} $\textsf{Sign}(1^{\ell}, sk_{\pi}, pk_1,\dots, pk_n,m)$, given a public key tuple \(\mathit{PK}_n\) $=(pk_1,\dots, pk_n)$, a pair \((sk_\pi, pk_\pi)\), where $\pi\in \{ 1,\dots,n\}$, starts by setting  
$$
            L = \mathit{H2P}(\mathit{HPK}(\mathit{PK}_n)), \,     T = sk_\pi \cdot L.
$$

Here, $\mathit{HPK}$ is a hashing procedure (that could actually be replaced by others, as in the generic construction of~\cite{DBLP:conf/acisp/LiuWW04}) defined for public key tuples  as follows. Given
$$
    \mathit{PK}_i = (pk_1, ..., pk_i) \qquad 0< i\leq n,
$$
we set ($||$ stands for the string concatenation operator),
$$
    \mathit{HPK}(\mathit{PK}_i) := 
  \begin{cases} 
   H(pk_1) & \text{if } i = 1 \\
   H(\mathit{HPK}(\mathit{PK}_{i-1})||pk_i)       & \text{if } i > 1.
  \end{cases}
$$

Let us denote by $x \in_R X $ the selection, uniformly at random, of an element $x$ from the set $X$. Values $s_1,\dots, s_n$ are computed, starting from index \(\pi\) as follows: 
$$
    u \in_R \{1,\dots, g-1\}, \,   s'_{\pi} \in_R \{1,\dots, g-1\} $$ $$
    A_{\pi} = s'_{\pi} \cdot G + u \cdot  pk_{\pi}, \,   B_{\pi} = s'_{\pi} \cdot L + u \cdot  T $$
    $$c_{\pi} = H(m || T ||A_{\pi}||B_{\pi})).$$

while for coefficients \(i = \pi+1, ..., n, 1, ..., \pi-1\):
$$     s_i \in_R \{1,\dots, g-1\}, \,    A_i = s_{i} \cdot G + c_{i-1} \cdot pk_{i} $$
$$     B_i = s_{i} \cdot L + c_{i-1} \cdot T, \,     c_i = H(m || T ||A_i||B_i)).$$

Finally, \(s_{\pi}\) is computed as $s_{\pi} = s'_{\pi} + sk_{\pi}(u - c_{\pi-1})$ and $c = c_n$ (All computations in the subindexes $i$  above are done \(mod \, p\)).
The output signature $\sigma$ is defined as
$\sigma = (\mathit{PK}_n, T, s_1,\dots, s_n, c).$

\medskip

\noindent{\bf Signature Verification.} $\textsf{Verify}(1^{\ell}, pk_1,\dots, pk_n,m,\sigma)$ first performs the following computations from the input values $m, \mathit{PK}_n$ and {$\sigma$}  
$$        L = \mathit{H2P}(\mathit{HPK}(\mathit{PK}_n)), \,  c_0 = c
$$.
Furthermore, it builds, for \(i = 1, ..., n\),
$$        A_i = s_{i} \cdot G + c_{i-1} \cdot pk_{i}, \,  B_i = s_{i} \cdot L + c_{i-1} \cdot T,\,$$
$$        c_i = H(m||T||A_i||B_i).
$$
Finally, if \(c_n\) equals \(c_0\), it outputs $1$. Otherwise, it outputs $0$ and rejects the signature as invalid.  

\medskip

\noindent{\bf Linkability Check.} $\textsf{Link}(\sigma_1,\sigma_2,m_1,m_2)$, given $$\sigma_1 = (\mathit{PK}_n, T_1, s_1,\dots, s_n, c), \,  \sigma_2 = (\mathit{PK}_n, T_2, \hat{s}_1,\dots, \hat{s}_n, \hat{c}),$$ outputs $1$ if an only if $T_1 = T_2 .$

\medskip

We are able to formally assess  the following result: 

\begin{theorem}  The ring signature scheme presented attains unforgeability, anonymity, linkability, and non-framability in the random oracle model\footnote{proofs in the \emph{random oracle model} work under the assumption that the hash functions involved behave as idealized/truly random functions (see~\cite{RSAOEP})} and under the DDH assumption in the underlying cyclic group $E$ generated by $G$. 
\end{theorem}

\subsection{Efficiency (storage)}
\label{sec:storageEfficiency}
Our choice of an ECC linkable ring signature scheme is due to  economy in keys and signatures sizes.
Assuming private keys of 256 bits (32 bytes), Table~\ref{tab:ec_size} summarizes the space required for public keys and signatures compared to the same implementation in a group of prime order (reaching the same security level as with ECC).

We also assume that the list of public keys $PK_n$ does not change once registration ends, so it can be saved once for all the signatures referring to the whole list. For the ECC case, it is useful for smart contract computations to have public keys (that are points in the curve) saved in an uncompressed form, so both coordinates are stored. 
Hence, for $n$ voters, the length of the signature is the tag $T$ length, $n$ times the ring coefficient $s_i$ length, plus the anchor coefficient $c_n$ length. For the ECC case, in bytes, the signature length is $64 + 32(n+1) = 32(n+3)$ (because $T$ is an elliptic curve point, while the coefficients have the same size as the private key); the public key is an elliptic curve point, so its length is $64 \cdot n = 32 \cdot 2n$ bytes. Using cyclic groups over finite fields (i.e., which we call non-ECC case), both $T$ and the public keys belong to the same group as the private keys, so each signature needs  a space of $384 + 384(n+1) = 384(n+2)$ bytes, while storing the ring of public keys requires $384 \cdot n$ bytes.

\begin{table}[h]
    \centering
    \caption{Storage requirements in bytes of Linkable Ring Signatures over Elliptic Curves (256 bit secret key) versus the original proposal from~\cite{DBLP:conf/acisp/LiuWW04} (3072 bit secret key)}
    \label{tab:ec_size}
    \begin{tabular}{|c|c|c|c|c|}
        \hline
        Signers & ECC PKs & ECC Sig & non-ECC PKs & non-ECC Sig \\
        \hline
        10 & 640 & 416 & 3840 & 4608 \\
        \hline
        100 & 6400 & 3296 & 38400 & 39168 \\
        \hline
        1000 & 64000 & 32096 & 384000 & 384768 \\
        \hline
    \end{tabular}
\end{table}
Public key's ring size in the ECC case is roughly one fifth of the non-ECC case, but the most important saving in storage comes for signature sizes: the ECC signature length is less than one tenth of the non-ECC case. As we will see in Section \ref{sec:voting-protocol}, a signature is saved in the blockchain per each ballot, so the saving in space with the ECC approach is huge.

\section{Voting framework}
\label{sec:voting-framework}
This section presents the design of the proposed system. We start by defining the different actors involved in its construction. Then, we will move to describe their interactions. In Appendix~\ref{appendix:voting-diagrams} we provide sequence diagrams that summarize the voting protocol.

\subsection{Actors}
The voting process is divided into phases, in which different actors play specific roles. Each actor is intended to be an independent entity for simplicity, but nothing prevents an actor from being a set of cooperating subjects, leveraging blockchain and smart contracts capabilities or, as an alternative, legal contracts or other constraints. We will mention related scenarios in the subsequent Section~\ref{sec:deployment}.

\paragraph{Voter}
\label{par:voter}
A voter is the entity engaged in the voting process with the aim to express his/her vote. Voters own enough information for being correctly identified (Line~\ref{alg:voter:iddata} Alg.~\ref{alg:voter}), as well as to store and manage secret data linked to their digital identity. It is their own responsibility to make correct computations following the scheme specification in order to generate and submit a valid ballot.

\begin{algorithm}[h]
\small
\caption{Code for Voter}
\label{alg:voter}
\begin{algorithmic}[1]
    \State \textbf{Init:} $IDData \leftarrow$ Voter Identification Data \label{alg:voter:iddata}
    \State \textbf{Init:} $(pk, sk) \leftarrow$ KeyGen()
    \Function{SignUp}{~} \label{alg:voter:signup}
        \State IdentityManager.{\sc SignUp}($IDData$, $pk$) \label{alg:voter:signupCall}
    \EndFunction
    \Function{Vote}{$v$}
        \State $PKS \leftarrow$ IDStorage.Get()
        \State $ballot \leftarrow$ Encrypt($v$, ConfManager.PublicKey) \label{alg:voter:encrypt}
        \State $\sigma \leftarrow$ Sign($sk$, $PKS$, $ballot$) \label{alg:voter:ringSign} \Comment{Computes the ring signature}
        \State BallotBox.{\sc Vote}($\sigma$, $ballot$) \label{alg:voter:send}
    \EndFunction
    \Function{GetResult}{~}
        \State \textbf{return} BallotBox.$Result$
    \EndFunction
\end{algorithmic}
\end{algorithm}
\paragraph{Organizer}
The \textit{Organizer} is the entity that configures and selects parameters characterizing a concrete voting event, i.e., vote's properties and voters' rights. It is responsible for setting up the overall voting session.

\begin{algorithm}[h]
\small
\caption{Code for Organizer}
\label{alg:organizer}
\begin{algorithmic}[1]
    \State \textbf{Init: deploy} ID Storage, Ballot Box, Confidentiality Manage \label{alg:organizer:init}
    \State \textbf{Init: configure} vote open, vote close \label{alg:organizer:configure}
    \Upon{ConfManager.$SecretKey \neq \emptyset$}
        \State $Ballots \leftarrow$ BallotBox.Get()
        \State $Result \leftarrow$ Decrypt($Ballots$, ConfManager.$SecretKey$)
        \State BallotBox.{\sc SetResult}($Result$) \label{alg:organizer:result}
    \EndUpon
\end{algorithmic}
\end{algorithm}

\paragraph{Identity Manager}
Identification is taken care of by the so-called \textit{Identity Manager}. Its main task is to identify actors involved in the voting process and to register them (Line~\ref{alg:idManager:verify} and \ref{alg:idManager:storageSignup}  Alg. \ref{alg:idManager}) in the platform according to the policies defined by the \textit{Organizer}. The $label$ derived from the $IDData$ works as identifier for the voter and can be useful to match the list of registered voters to some external public one.

\begin{algorithm}[h!]
\small
\caption{Code for Identity Manager}
\label{alg:idManager}
\begin{algorithmic}[1]
    \Function{SignUp}{$\mathit{IDData}$, $pk$}
        \If{VerifyIdentity($\mathit{IDData}$)} \label{alg:idManager:verify}
            \State $label \leftarrow$ LabelFromIDData($\mathit{IDData}$)
            \State \textbf{call} IDStorage.{\sc SignUp}($pk$, $label$) \label{alg:idManager:storageSignup}
        \EndIf
    \EndFunction
\end{algorithmic}
\end{algorithm}

\paragraph{Confidentiality Manager}
When required, a \textit{Confidentiality Manager} is responsible to guard and distribute information necessary to encrypt and decrypt votes. This is a sensitive component, and Section~\ref{sec:confidentiality} goes in depth about  possible solutions that guarantee confidentiality with and without a \textit{Confidentiality Manager}.
\begin{algorithm}[h]
\small
\caption{Smart Contract for Confidentiality Manager}
\label{alg:confidentialityManager}
\begin{algorithmic}[1]
    \State \textbf{Init:} $\mathit{PublicKey} \leftarrow DistrPKGeneration()$ \label{alg:confidentialityManager:pkGen}
    \State \textbf{Init:} $\mathit{SecretKey} \leftarrow \emptyset$
    \Upon{vote close} \Comment{event fired/enabled by the blockchain}
        \State $\mathit{SecretKey} \leftarrow DistrSKGeneration()$ \label{alg:confidentialityManager:skGen}
    \EndUpon
\end{algorithmic}
\end{algorithm}

\subsection{Voting Protocol}
\label{sec:voting-protocol}
To describe the voting process we introduce two more components:
\begin{itemize}
    \item the \emph{ID Storage} (see Alg.~\ref{alg:idStorage}) is the bulletin board of eligible voters.
    \item the \emph{Ballot Box} (see Alg.~\ref{alg:ballotBox}) is the storage of valid ballots.
\end{itemize}
\begin{algorithm}[h]
\small
\caption{Smart Contract for ID Storage.}
\label{alg:idStorage}
\begin{algorithmic}[1]
    \State \textbf{Init:} $PublicKeys \leftarrow \emptyset$
    \Function{SignUp}{pk, label} \Comment{Only by the Identity Manager}
    \If{$\neg$ BallotBox.$active$}
        \State $PublicKeys \leftarrow PublicKeys || {<pk, label>}$ \label{alg:idStorage:store}
    \EndIf
    \EndFunction
\end{algorithmic}
\end{algorithm}

\begin{algorithm}[h]
\small
\caption{Smart Contract for Ballot Box. $p$ is the total number of possible choices for a ballot.}
\label{alg:ballotBox}
\begin{algorithmic}[1]
    \State \textbf{Init:} $choices \leftarrow \{v_0, ..., v_p\}$ \label{alg:ballotBox:encoding}
    \State \textbf{Init:} $active \leftarrow false$
    \State \textbf{Init:} $Ballots \leftarrow \emptyset$; $Tags \leftarrow \emptyset$; $Result \leftarrow \emptyset$
    \Function{Vote}{$\sigma, b$}
        \If{$\sigma.T \notin Tags \land active$} \Comment{The scheme Link function}
            \If{Verify(BallotBox.$PublicKeys, b, \sigma)$}
                \State $Ballots \leftarrow Ballots \cup {b}$
            \EndIf
        \EndIf
    \EndFunction
    \Upon{vote open} \Comment{event fired/enabled by the blockchain}
        \State $active \leftarrow true$ \label{alg:ballotBox:enable}
    \EndUpon
    \Upon{vote close} \Comment{event fired/enabled by the blockchain}
        \State $active \leftarrow false$
    \EndUpon
    \Function{SetResult}{$R$} \Comment{Only by Organizer}
        \State $Result \leftarrow R$
    \EndFunction
\end{algorithmic}
\end{algorithm}

With the above ingredients, the framework phases are described as follows: 

\subsubsection{Setup}
The \textit{Organizer} deploys and configures on the blockchain the three smart contracts that hold the public information for the voting session: \textit{ID Storage}, \textit{Ballot Box}, \textit{Confidentiality Manager} (Line~\ref{alg:organizer:init} of Alg.~\ref{alg:organizer}). The \textit{Organizer} configures, in this phase, the firing for the \textit{vote opening} and \textit{vote closing} events. Moreover, the \textit{Organizer} declares the tolerance to failures for the \textit{Identity Manager} and the \textit{Confidentiality Manager} if they adopt any threshold mechanism ($K$ out of $N$ discussed in Sections~\ref{sec:identification} and \ref{sec:confidentiality}).

\subsubsection{Registration}
Voters are required to sign up (Line~\ref{alg:voter:signup} Alg.~\ref{alg:voter}) to the system by providing (a) information that allows to prove their identity and (b) the public key they are going to use for the vote (Line~\ref{alg:voter:signupCall} Alg.~\ref{alg:voter}). This is a simplification of a typical certificate issuing procedure, since the outcome is not a signed certificate but a registration on the ID Storage contract.
Only voters who sign up strictly before the vote opening will be allowed to vote.

\subsubsection{Vote Opening}
The \textit{Ballot Box} is set to active (Line~\ref{alg:ballotBox:enable} of Alg.~\ref{alg:ballotBox}) so it begins to accept ballots and the \textit{ID Storage} stops accepting new public keys. The vote opening/closing event reported in Algorithms~\ref{alg:organizer},~\ref{alg:idStorage}~and~\ref{alg:ballotBox} is pointed as an external condition, since it is enforced by leveraging the blockchain and smart contract property of referring to the constant-time block issuance process. In this way, operations are constrained in respect to time. Anyway, if this property is not available in the chosen system deployment, the \textit{Organizer} is in charge of firing those events by writing them in the blockchain.

\subsubsection{Voting}
Eligible voters interact with the platform in order to get the information required for voting. In order to submit a valid ballot, voters: (1) fetch identities needed for the linkable ring signature from the \emph{ID Storage} (Line~\ref{alg:voter:iddata} Alg.~\ref{alg:voter}), (2) use the format/encoding defined in the \emph{Ballot Box} to create a ballot, encrypt it by fetching information from the \textit{Confidentiality Manager}~(Line~\ref{alg:voter:encrypt} Alg.~\ref{alg:voter}), (3) compute the linkable ring signature (Line~\ref{alg:voter:ringSign} Alg.~\ref{alg:voter}) and finally (4) anonymously send their vote to the system (Line~\ref{alg:voter:send} Alg.~\ref{alg:voter}).
\subsubsection{Vote Closing}
The same considerations made for the Vote Opening hold for the Closing event.

\subsubsection{Tally}
Once the Voting Phase has been closed, the \textit{Confidentiality Manager} is in charge of releasing the Secret Key paired with the Public Key published for ballot encryption (Line~\ref{alg:confidentialityManager:skGen} Alg.~\ref{alg:confidentialityManager}). Once done, anyone, thanks to the information publicly available on the blockchain, is able to decrypt ballots and compute the final result.

\section{Protocol details and design choices}
\label{sec:deployment}
This section gives details on how our framework achieves  \textit{Identification}, \textit{Confidentiality}, \textit{Anonymity} and \textit{Scalability} in different application scenarios.

\subsection{Identification}
\label{sec:identification}
Identification is a vital step in the vote startup. There are many known and well assessed procedures employed nowadays for mapping real identities to digital ones, like adopting multiple factors of authentication and/or relying on certificate authorities. We deliberately wrap any of these procedures in the $IDData$ variable (Line~\ref{alg:voter:iddata} Alg.~\ref{alg:voter}). The \textit{Identity Manager} is thus responsible for mapping the chosen identification procedures to our system. Since achieving resilience to Byzantine failures (unwanted or malicious failures and deliberate misbehavior) is one of our major goals, it is recommended to distribute the identification responsibility among a set of entities. That is,  multiple entities should contribute to the \textit{Identity Manager}'s operation and  can leverage the smart contract capabilities in order to cooperate (e.g., use a $K$ out of $N$ threshold mechanism).

Moreover, to let failures be as independent as possible, the \textit{Identity Manager} should implement different voters' identification procedures (that is, in our construction, implement different VerifyIdentity($IDData$) functions in Line~\ref{alg:idManager:verify} of Alg.~\ref{alg:idManager}).

\subsection{Confidentiality}
\label{sec:confidentiality}

Our system guarantees secrecy to voters and avoids unwanted early result disclosures. To this aim, there are several concrete solutions that may be implemented. Indeed, a \textit{Ballot box} built on the blockchain can be compared to a transparent box with unfolded paper inside and anyone being able to take note of added votes. The vote result is known step by step during the Voting Phase and, in most use cases, this is not a desired possibility because real time results may influence the voters who haven't voted yet. That is why a mean to obfuscate votes contents during the Voting Phase is required.

\subsubsection{Ballot encoding and encryption}
\label{sec:encodingAndEncryption}

Each participating entity can acquire information about choices encoding from the \textit{Ballot Box} (Line~\ref{alg:ballotBox:encoding} Alg.~\ref{alg:ballotBox}). Then, we rely on the fact that the Encrypt function in Line~\ref{alg:voter:encrypt} of Alg.~\ref{alg:voter} is always able to provide a different ciphertext at each execution even if encrypting the same plaintext. If encryption were deterministic, this is trivially achieved by embedding a random salt at each encryption (non-deterministic schemes, like RSA-OAEP~\cite{RSAOEP}, will do the job without modifications).

\subsubsection{Confidentiality Manager}
\label{sec:confidentiality-manager}

The \textit{Confidentiality Manager} is the component in charge of providing the public key used to encrypt ballots during the Voting Phase. From the trustworthiness perspective, the \textit{Confidentiality Manager} has the possibility to completely invalidate a voting session because the tally directly depends on the secret key disclosure. Hence, for an attacker, it would be sufficient to compromise the \textit{Confidentiality Manager} to void the result of the voting session. For this reason, the responsibility of generating a key pair at the correct time is shared among a set of cooperating entities. Properly designed for use cases like ours, the Ethereum Distributed Key Generation (ETHDKG) \cite{DBLP:journals/iacr/SchindlerJSW19} allows different entities to generate and share a key pair with a threshold mechanism of $K$ out of $N$ for private key disclosure. It exploits Shamir's Secret Sharing \cite{DBLP:journals/cacm/Shamir79} and Verifiable Secret Sharing \cite{DBLP:conf/focs/Feldman87}, and relies on the blockchain as an auditable communication mean. In particular the whole ETHDKG protocol is executed before vote opening, resulting in the master public key for ballot encryption (Line~\ref{alg:confidentialityManager:pkGen} Alg.~\ref{alg:confidentialityManager}). As soon as the vote is closed, the individual secret key shares are  revealed in order to execute the Tally Phase (Line~\ref{alg:confidentialityManager:skGen} Alg.~\ref{alg:confidentialityManager}). If $K$ is the number of shares required to reconstruct the secret key and $N$ the number of \textit{Confidentiality Manager} members, the tolerance to malicious behavior is $min(K-1,N-K)$; in fact the system correctness can be undermined in two cases: (1) the secret is disclosed before the Tally Phase, so $K$ entities disclosed earlier; or (2) the secret is never disclosed because $N-K+1$ entities do not correctly participate in the secret reconstruction. Since both cases must be avoided, tolerance to failing entities is the minimum between the two.

\subsubsection{Vote Claim}
\label{sec:vote-claim}

As an alternative to the distributed key generation, a two step procedure is possible. (\textit{Submit Step}) The voters do not encrypt their ballots and they do not submit the ballot content in the voting phase. Instead, they (1) salt the ballot content, (2) compute a hash of it, (3) ring-sign the hash string, and (4) they send to the \textit{Ballot Box} the ring signatures along with the hash string. If the signature is valid, then both, signature and hash string, are immutably stored in the blockchain and, since ballots are salted, there is no feasible way for a computationally bounded adversary to invert the hash and derive the original ballot. When the tally begins (\textit{Redeem Step}), voters can

submit their ballots to the smart contract which will accept only those with a matching hash stored in the previous step.

Although this seems to solve the confidentiality problem there are some important drawbacks to take into account:
\begin{itemize}
    \item The voters must interact one more time with the platform to claim their vote in the Tally Phase, and this can be an issue for effective participation with a large number of voters (i.e., if a ``participation threshold'' exists, it should be reached not only during the Voting Phase but also during the tally interaction).
    \item The voters need to store their ballot somewhere, waiting for the Voting Phase to finish. The device or system they rely on should be kept secure during this time.
    \item During the Tally Phase results are collected in a time window, so a voter could decide to submit or not his/her ballot depending on the ongoing result. In some cases, this would void the neutrality of the system.
\end{itemize}
If the voting use case does not need to be resilient to these issues, this is an effective solution to guarantee confidentiality.

\subsection{Anonymity}
\label{sec:anonymity}

Ring signatures grant anonymity, providing at the same time the required authenticity guarantees. This solves the problem of recognizing a ballot as valid without disclosing the underlying voter identity, thanks to the fact that a ring signature obfuscates the link between a signature and the related signer.

On the blockchain interaction side, in Line~\ref{alg:voter:send} of Alg.~\ref{alg:voter} the voter directly calls a method of a smart contract that causes a state change in the blockchain; unfortunately this is not usually possible for several reasons. First, a voter may not have enough storage or computation resources to run a full or light blockchain client. This can be easily overcome by sending transactions to a set of nodes acting as mediators for the submission of end-users' transactions. A second, more relevant, issue is that a blockchain user must sign his/her own transaction with a valid blockchain identity. The word ``valid" can have, in this case, two different meanings: in a \emph{permissioned} blockchain it means that the identity is authorized to submit transactions, while in a \emph{permissionless} one it means that the identity has enough resources to create an acceptable transaction. Neither condition pairs well with an anonymous and scalable ballot submission. The former would mean to identify the actual user behind the identity, so voiding the linkable ring signature property. The latter would imply the voter to procure by himself/herself the needed resources, which is usually not scalable on a large non-technical voter base. Moreover, the resource (cryptocurrency) procurement should preserve anonymity. We present three ways to solve this problem.

\subsubsection{Anonymization Proxy}
\label{sec:anon-proxy}

The anonymization proxy is a component that acts as a proxy between voters and the blockchain. In practice it is responsible of taking valid votes, packing them in transactions, and submitting them to the blockchain. Multiple authorized anonymization proxies can be provided to avoid a single point of failure and trust. Voters can autonomously choose which (one or more) they prefer and they can send the ballot without being identified. If they require advanced anonymity guarantees, they can leverage any anonymization technique for ordinary web interactions (e.g. Onion routing). After the submission, they can easily check that their operations have been effectively executed on the blockchain, thanks to its transparency. On the \textit{Organizer}'s side this has advantages, since there is no need to ensure write access to the blockchain to each voter, but only to the authorized proxies, who can also pack ballots in less transactions and optimize execution costs.

\subsubsection{Application Specific Blockchain}

An application specific blockchain can be realised with an ad-hoc designed execution environment. This means that the application logic of the voting system can be directly embedded in the transactions semantics. Mapping to our case, it is possible to validate the transactions of ballots' submissions directly through the attached linkable ring signature. The Tendermint SDK \cite{buchman2016tendermint} is an example framework on which such implementation could be based.

\subsubsection{Vote Token}
\label{sec:vote-token}

There can be some scenarios in which requiring a physical interaction is not a problem (i.e., a plenary meeting of association members). In this case it is possible to think of some vote token to be randomly pulled out of a box by each voter, who can then use it to claim the resources (cryptocurrency) necessary to submit a vote to the system, using a throwaway address generated on the fly. Furthermore, tokens can be accumulated for future votes.

\subsection{Scalability}
\label{sec:scalability}
The voting system should easily scale in the presence of a large base of actors. In particular, the number of voters is expected to be orders of magnitude higher than the number of all the other participating entities.

Since scalability currently is a key issue in blockchain-based systems,  we take special care dealing with it and furthermore compare our design with similar proposals. Our metric will be in terms of required communication rounds ($c$), and read ($r$) and write ($w$) operations on the blockchain for each Voting Phase. When several information pieces are available to be written/read, they are assumed to be accessed in a single operation. Among all possible deployment choices, we will consider a setup similar to the two related works based on blockchain \cite{DBLP:conf/trustcom/LyuJWNAF19, DBLP:journals/iacr/LiuW17}: we assume to have $M$ voters and to impose $K_i$ on $N_i$ threshold mechanisms for voters identification and $K_e$ on $N_e$ secret sharing threshold for ballot encryption. Table~\ref{tab:scalability} shows a comparison with these two proposals \cite{DBLP:journals/iacr/LiuW17, DBLP:conf/trustcom/LyuJWNAF19}. The former is based on blind signatures and the latter, like ours, on linkable ring signatures.

In the Registration Phase, voters must successfully identify themselves with at least $K_i$ different \textit{identity managers}; so, a voter has at least a communication round with each one of them. Unlike \cite{DBLP:journals/iacr/LiuW17}, this process can be reused indefinitely for any subsequent voting session, even if not initially scheduled. Any subsequent vote can directly start its protocol considering this phase as already completed. This is really efficient for recurring votes with a great number of voters in common.

Then, voters read the ring of public keys and the encryption key from the blockchain, and they submit their ballot along with the ring signature; this translates into $1$ read and $1$ write in the blockchain per voter. Providing confidentiality through the ETHDKG \cite{DBLP:journals/iacr/SchindlerJSW19} protocol requires 2 rounds of interactions by the $N_e$ \textit{Confidentiality Managers}, in which each of them has to write to the blockchain in both rounds (\textit{Encryption}) in order to publish the \textit{master public key}. In the Tally Phase it is enough that $K_e$ \textit{Confidentiality managers} disclose their secrets to decrypt ballots.

\begin{table}
\caption{Summary table of reads $r$ and writes $w$ on the blockchain globally required by the system in each phase in $c$ communication rounds. Note that $M, K_m$ are usually much greater than $N_e, K_e$}
    \label{tab:scalability}
    \resizebox{\columnwidth}{!}{
        \begin{tabular}{|c|c|c|c|c|c|c|c|c|c|}
        \hline
        Phase &
      \multicolumn{3}{c|}{Our Framework} &
      \multicolumn{3}{c|}{Blind Sig. \cite{DBLP:journals/iacr/LiuW17}} &
      \multicolumn{3}{c|}{Alt. LRS \cite{DBLP:conf/trustcom/LyuJWNAF19}} \\
    & w & r & c & w & r & c & w & r & c \\
    \hline
    Setup & $1$ & - & $1$ & $1$ & - & $1$ & $1$ & - & $1$ \\
    \hline
    Registration & $MK_i$ & - & $K_i$ & \multicolumn{3}{c|}{Not required} & \multicolumn{3}{c|}{Not provided} \\
    \hline
    Encryption & $2N_e$ & $N_e$ & $2$ & \multicolumn{3}{c|}{Not provided} & $2M$ & $M$ & $2$ \\
    \hline
    Voting & $M$ & $M$ & $1$ & $2MK_i$ & $2MK_i$ & $2K_i$ & $M$ & $M$ & $1$ \\
    \hline
    Tally & $K_e$ & - & $1$ & - & $1$ & $1$ & $K_m$ & - & $1$ \\
    \hline
  \end{tabular}
}

\end{table}

For the \textit{vote claim} solution of Section~\ref{sec:confidentiality}, \textit{Confidentiality Managers} are not required, and ballot encryption/decryption is completely eliminated at the cost of a second communication round for the voter in the Tally Phase. In that case, Tally Phase cost becomes $M$ writes and no distributed key generation protocol is required, avoiding the \textit{Encryption} and \textit{Decryption} costs.

\section{Properties and additional features}
\label{sec:requirements}

We now  discuss relevant additional properties achieved by our design, which include and extend those reported in~\cite{DBLP:journals/iacr/LiuW17}. For each one of them we provide an informal definition and discuss why we are able to fulfill it.

\subsubsection{Required properties}
Required properties are those that must be granted for any use case of the e-voting system, as they encompass the minimal correctness and security prerequisites needed for any application scenario:

\paragraph{Verifiability}
The \textit{ID Storage} (Alg.~\ref{alg:idStorage}) and the \textit{Ballot Box} are implemented by smart contracts, leveraging the blockchain capabilities. That is, each method invocation is permanently stored in the system. All ballots remain saved and publicly visible along with their attached linkable ring signature. So, a voter, who knows which his/her own ballot is, can verify the correctness of both the ballot submission and the tally (Individual Verifiability). From the global system point of view (Universal Verifiability), every entity can read the \textit{ID Storage} and verify the linkable ring signature, checking that all saved ballots are legitimate. Finally, the saved result (Line~\ref{alg:organizer:result} Alg.~\ref{alg:organizer}) can be verified through the secret key published by the \textit{Confidentiality Manager}. The \textit{Organizer} saves back the result to the \textit{Ballot Box} only as a convenience, since the secret key for decryption is already public at that moment.

\paragraph{Legitimacy}
Only those who have the right to vote can take part in the voting process, and only for as long as they're granted to participate. The satisfaction of this property is directly tied to the \textit{Identity Manager} correctness. That's why a $K$ out of $N$ mechanism is recommended. Each information in the \textit{ID Storage} is provided by the (or $K$ out of $N$) \textit{Identity Manager} through signed transaction submitted to the blockchain. The label, saved along with the public key (Line~\ref{alg:idStorage:store}, Alg.~\ref{alg:idStorage}), acts as a public identifier for each voter in the list, so anyone who knows the list of allowed voters, can check it against the registered one (e.g., in our deployment in Section~\ref{sec:deployment}, we set the label as the SHA256 of the voter's email). Beyond this, the linkable ring signature, paired with the deterministic execution environment of the blockchain, ensures that only identities registered in the \textit{ID Storage} can create and successfully submit valid ballots. In fact, signature verification is implemented in the \textit{Ballot Box} smart contract, so double votes are detected through the linkability and non-frameability properties.

\paragraph{Completeness} 
All ballots are correctly interpreted and vote accounting is accurate. In our design, when ballot secrecy is not required (but only anonymity is), then the Tally Phase consists of a simple count function over publicly available data. The same result is achieved with the \textit{vote claim} solution in Par.~\ref{sec:vote-claim} after the \textit{redeem} step. Note that in this case there is the possibility that included fingerprints (stored along with the ring signatures) are not redeemed, resulting in uncounted ballots. Depending on the use case, this is acceptable or not.

When using \textit{Confidentiality Manager}s that cooperate by leveraging ETHDKG \cite{DBLP:journals/iacr/SchindlerJSW19}, completeness is guaranteed if the distributed key generation algorithm is correctly followed. Tolerance to failures is analysed in Section~\ref{sec:confidentiality-manager}.

\paragraph{Neutrality}
The voting  outcome cannot be influenced neither by the voting process itself nor by the tools used to vote. In the Confidentiality Section~\ref{sec:confidentiality}, we explain how our design allows for the vote result to be hidden during a voting session, in order to avoid influencing voters with already accepted ballots.

\paragraph{Auditability}
The whole voting procedure must be auditable during and after the vote. In our case, this feature is provided by the blockchain, which gives the abstraction of an append-only register that allows auditing all methods' invocations on the components implemented through smart contracts (\textit{ID Storage}, \textit{Ballot Box}, \textit{Confidentiality Manager}).

\paragraph{Consistency}
The voting result is the same for all correct actors and stakeholders. Correct voters submit ballots using the encoding rules defined in the \textit{Ballot Box} (Line~\ref{alg:ballotBox:encoding} Alg.~\ref{alg:ballotBox}) smart contract. From the system point of view, the deterministic execution layer of the blockchain, paired with its consensus algorithm, ensures that all nodes execute transactions in the exact same order, hence constructing exactly the same system state for all involved smart contracts.

\subsubsection{Information Security}
Standard security properties, in general referred as the CIA triad \cite{ISO27000}: confidentiality, integrity, and availability. We will refer to data confidentiality, data integrity, system integrity, and availability as defined in \cite{RFC4949} and briefly comment on how our framework has been designed to achieve them:

\paragraph{Data Confidentiality}
Data  should not be disclosed to system entities unless they have been authorized to access it. In Section~\ref{sec:confidentiality} we discussed about confidentiality of ballots as the possibility to not disclose them during the Voting Phase. In respect to voters we also analysed a solution for their anonymity in Section~\ref{sec:anonymity} beyond the mere theoretical usage of the linkable ring signature.

\paragraph{Data Integrity}
Data  should not be changed, destroyed, or lost in an unauthorized or accidental manner. As long as the Byzantine fault tolerant consensus algorithm of the blockchain remains valid, all the constraints coded in the e-voting session smart contract are respected and no one can tamper with them. Linkable ring signature verification is coded inside the Ballot box smart contract.
\paragraph{System Integrity}
The system should perform its intended function in a unimpaired manner, free from deliberate or inadvertent unauthorized manipulation. Note that each change in  system state happens through signed transactions that execute deterministic pieces of code. So any modification is authenticated and propagates the same changes to all nodes.
\paragraph{Availability}
The system should be accessible, or usable or operational upon demand, by an authorized system entity, according to its performance specifications. Due to the distributed nature of the blockchain, the system remains available as long as the underlying consensus algorithm is correctly run by a subset of nodes; the minimum number of nodes in the subset is defined by the blockchain.

Some issues with availability can arise when allowing ballot replacement as a coercion countermeasure (see below). In fact, since each legitimate signer would be allowed to submit as many ballots as he/she wants, then all submitted ballots should be accepted. The problem resides in the fact that creating fake linkable ring signatures (by using random coefficients) is much more efficient than verifying them. For this reason an attacker could undermine the system's availability by flooding it with many fake signatures. Anyway our approach based on the \textit{anonymization proxy} (Section~\ref{sec:anon-proxy}) can apply known Denial-of-Service countermeasures for protecting endpoints while the \textit{vote token} approach (Section~\ref{sec:vote-token}), involving a physical action for each submission, is a natural mitigation.

\subsubsection{Additional Features}
The following features can either be required or not by different kinds of voting systems. Still, it is desirable that a voting system is able to provide them.

\paragraph{Archive}
Vote results should be stored with proved integrity. Each blockchain node holds a complete copy of data and history of operation on those data; thus integrity is guaranteed by the blockchain.

\paragraph{Coercion resistance}
The system should protect voters from being deprived of their rights. In particular, it should prevent:
\begin{itemize}
    \item Simulation: it should not be possible to impersonate a legitimate voter;
    \item Forced absence: voters may not be denied to vote;
    \item Coerced vote: voters may not be forced to vote in a predefined way.
\end{itemize}

For our design, we have that:
\begin{itemize}
    \item Simulation: for this, an attacker should either steal the private key of a voter after the vote phase started or make the \textit{Identity Manager} replace the key from the right voter with his/her own public key. In the first case we assume that a voter is able to correctly store his/her secret information; hardware tokens or trusted execution environments can be useful for this purpose, too. For the second case, a voter can detect that his/her public key is wrong by directly checking on the blockchain, so he/she can rerun the identification process. Moreover, several \textit{Identity Managers} can be present, each one possibly owning a different identification process, making it even harder for an attacker to succeed. 
    \item Forced absence: access to the blockchain is assumed ubiquitous. Moreover, voting steps can be automated using scripts; thus users may prepare votes and set up thier submission for a specific (later) time. Anonymity of the linkable ring signature ensures a voter that no one can link back to the ballot he/she submitted. Finally, non-framability does not allow any adversary, even if formed by a collusion of other legitimate signers, to deny vote of a honest user by forging his/her link tag $T$.
    \item Coerced vote: voters may be enabled to submit valid ballots an undefined number of times, as linkability tags allow for enforcing subsequent ballots to overwrite previous ones. Moreover voters may schedule some automatic vote submission right before the Voting Phase deadline. Anyway, we consider this workaround only a mitigation, since this legitimate double ballot submission can be easily detected by a coercer, given the auditability of the blockchain system, if he/she was able to take note of the linkability tag submitted by the coerced voter.
\end{itemize}

\section{Implementation}\label{sec:proof-of-concept}
We validate the  feasibility of our proposal with a working implementation. Ring signatures and the blockchain are points of paramount importance, so our proof of concept implementation mainly focuses on them.

\subsection{Smart Contract}
The most natural choice for rapidly prototyping a working concept based on a blockchain is to rely on the Ethereum ecosystem of software.
The \textit{ID Storage} and the \textit{Ballot Box} are implemented in a single smart contract that stores the voting session state, voters information, and ballot data. The smart contract also implements linkable ring signature verification, so only ballots with valid signatures are actually stored.
Signature algorithms reside in a separate library so they can be reused across different deployments of the system. The \textit{Organizer}, \textit{Confidentiality Managers} and \textit{Identity Managers} appear with their Ethereum addresses, they can hence be either individuals or other smart contracts like in the case of \cite{DBLP:journals/iacr/SchindlerJSW19}.

\subsection{Signature Algorithm}

The signature scheme described in Section~\ref{sec:lsag-schema} has been implemented in Solidity within the Ballot Box leveraging the BN256 curve since the Ethereum Virtual Machine has opcodes for the sum and multiplication on it. \cite{eip1108}.\footnote{With an application-specific blockchain there would not have been any constraint in reference to this, indeed the execution layer may be customized as needed.}).

Moreover, in our schema, the public keys hash gets computed cumulatively at at each new added public key. This particularly suits the blockchain environment since the computation is spread across multiple transactions.

\subsection{Proof of concept}
We implemented the whole voting process by providing scripts to accomplish each phase we described for each involved actor. Confidentiality is provided through the ETHDKG \cite{DBLP:journals/iacr/SchindlerJSW19} smart contract. The codebase \cite{chitoronia-poc} is in Python, except for the smart contracts that are written in Solidity, the smart contract oriented programming language of the Ethereum blockchain.

The scripts can target any Web3 (the Ethereum standard RPC interface) compliant node. We also provide a ready-to-go testnet through a Docker container.

\subsection{Real deployment}
Our voting framework has been successfully deployed and used in the elections of university officers. The three distinct elections involved 428 total voters with a participation of 408 voters in the registration phase, and 398 in the voting phase. Although we are going to produce and publish a full report of this first use of the system (and of the subsequent already scheduled ones), we think it is worth citing here the main characteristics of the event:
\begin{itemize}
    \item The overall architecture was made of a blockchain, a blockchain explorer, a web server acting as \textit{Identity Manager}, a web server acting as \textit{anonymization proxy} for vote submission, an administration portal and a web application for voters to interact with the system.
    \item The signature algorithm has been implemented in plain Javascript to compute the signature in a common browser.
    \item Voters generated a \textit{Voting Card} locally in their browser. The \textit{Voting Card} included a keypair, the voter email and a signature of the email, as a string, using the generated keypair. The secret key was encrypted with a password and the public key was sent to the \textit{Identity Manager}. The identity was verified by integrating our system with the university ones through SAML 2.0. Voters then downloaded the \textit{Voting Card} from the browser in order to store it on their device.
    \item We deployed a permissioned blockchain with Hyperledger Besu that is fully Ethereum compliant. We ran 4 nodes with the Istanbul Byzantine Fault Tolerant consensus algorithm \cite{ibft} (a tweak of \cite{DBLP:conf/osdi/CastroL99}). Each node was deployed and administered by a distinct entity.
    \item During the voting phase, voters selected their chosen candidate, uploaded the \textit{Voting Card} to the browser and inserted the password to decrypt the secret key. The browser computed the linkable ring signature and encrypted the ballot choice, and sent both to the \textit{anonymization proxy} that redirected the signature and ballot to the blockchain. The transaction hash and a link to the blockchain explorer was given back to the voter as receipt.
    \item We did not integrate the Ethereum DKG\cite{DBLP:journals/iacr/SchindlerJSW19}, yet, so the administration portal generated an RSA key pair and published the public key to the \textit{Ballot Box}. Once the voting phase was over, the members of the electoral committee requested ballots decryption on the administration portal. Decryption was computed as soon as all the members of the committee asked the decryption with their personal accounts.
\end{itemize}

\section{Conclusion}\label{sec:conclusion}

The presented framework leverages the blockchain and linkable ring signature capabilities to create a comprehensive solution for e-voting in a wide range of scenarios. We show how our design can address the requirements of a voting system and provide a collection of additional features. In particular our proposal fits voting scenarios where coercion resistance and receipt freeness are not a key concern. Indeed, our system cannot be considered receipt-free because of its individual verifiability property. Moreover, even though simulation and forced absence are to some extent prevented, coercions are not fully remedied, since a fully remote workflow can't prevent a voter to be physically conditioned.

Our plan is to integrate the Ethereum DKG to guarantee confidentiality and neutrality to the platform, and to make a step further in the scalability of the framework, by accomodating the system to support a larger number of voters (thousands). This, in particular, will allow us to tune up deployment choices for different election scenarios.

\section{Acknowledgement}
We are grateful to Javier Herranz for fruitful discussions and pointers to relevant literature on ring signature schemes.
This work is partially supported by the NATO Science for Peace and Security program (grant G5448), by the Spanish Ministerio de Ciencia e Innovación (MICINN) grant PID2019-109379RB-I00, and by the Regional Government of Madrid (CM) grant EdgeData-CM (P2018/TCS4499) cofounded by FSE \& FEDER. This paper is part of the R\&D\&I project PID2019-109805RB-I00 funded by MCIN/AEI/10.13039/501100011033 (ECID).

\printbibliography

\appendix

\subsection{Voting workflow diagrams}
\label{appendix:voting-diagrams}
We report the sequence diagrams that summarize the voting protocol following the algorithms in Section~\ref{sec:voting-protocol} in order to give an overall picture of the interactions among the actors of the framework.
\begin{figure}
\centering
\includegraphics[width=6.5cm]{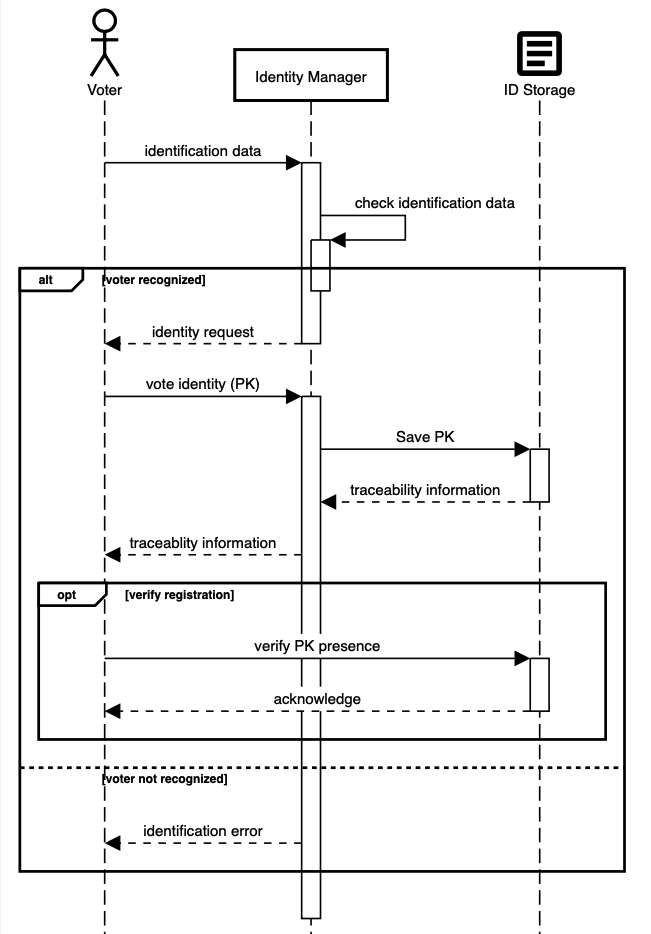}
\caption{Sequence diagram of the registration phase}
\label{seq:registrazione}
\end{figure}

\begin{figure}
\centering
\includegraphics[width=6.5cm]{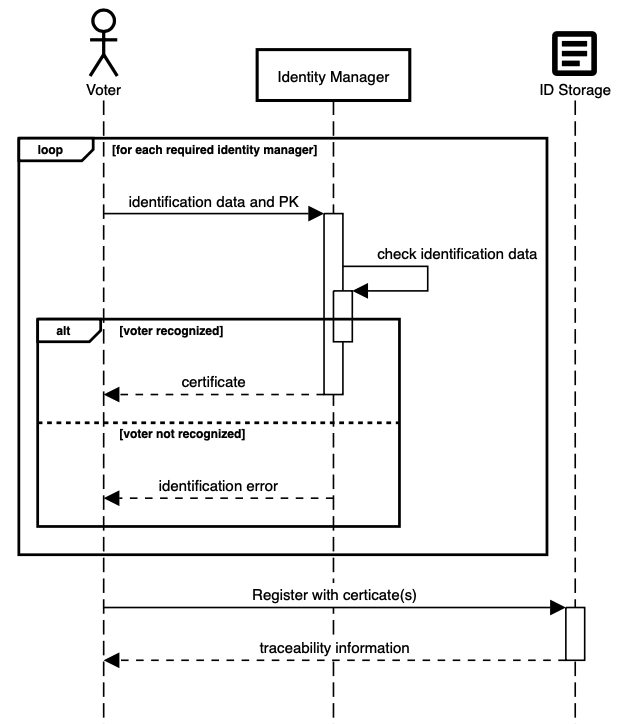}
\caption{Sequence diagram of the registration phase with multiple identity managers and using certificates.}
\label{seq:registrazione-opt}
\end{figure}

\begin{figure}
\centering
\includegraphics[width=6.5cm]{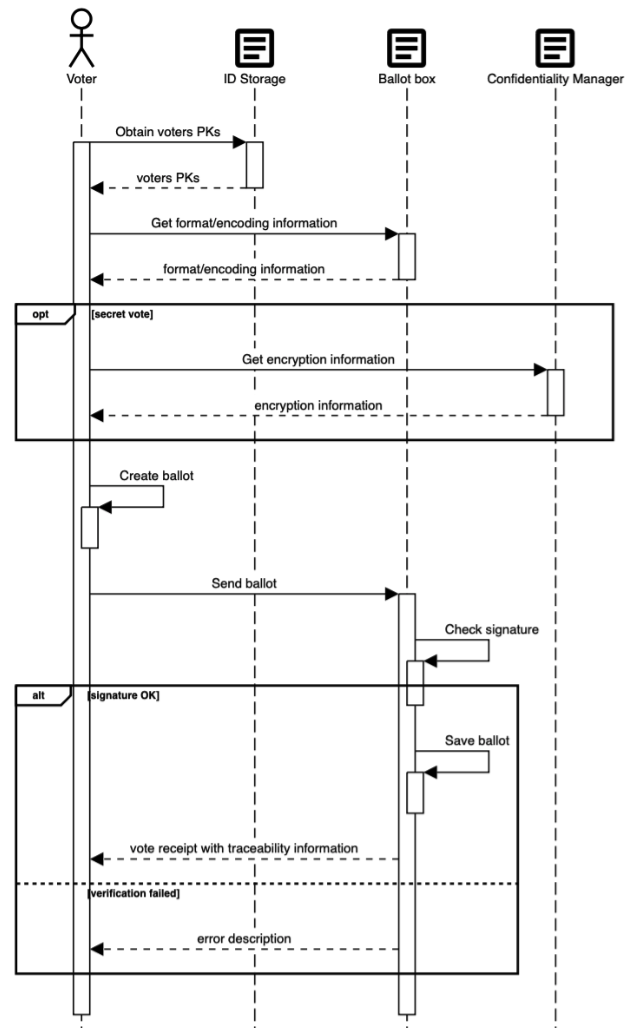}
\caption{Sequence diagram of the voting phase}
\label{seq:voto}
\end{figure}

\begin{figure}
\centering
\includegraphics[width=6.5cm]{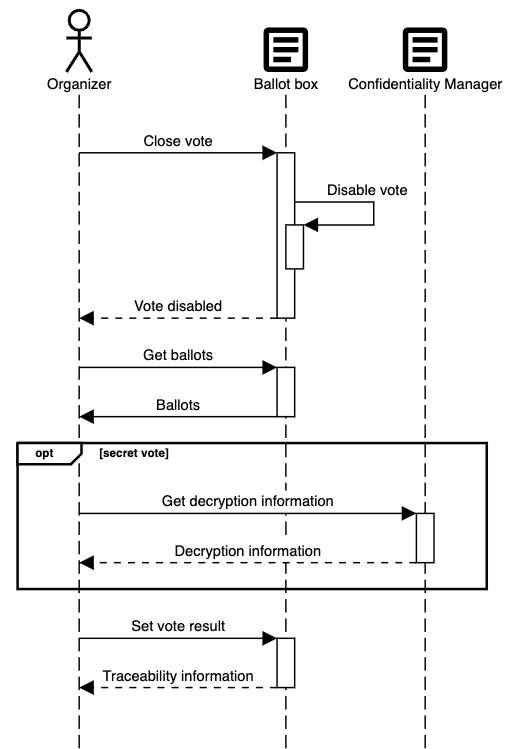}
\caption{Sequence diagram of the tally phase. The ``Close vote" message can be substituted by a blockchain automatic lock, coded in the ``Ballot box" smart contract.}
\label{seq:scrutinio}
\end{figure}

\subsection{Security Definitions for Linkable Ring Signature Schemes}

\label{appendix:ec-lsag}

We  start by providing formal definitions of the main security notions for linkable ring signatures as defined in Section~\ref{sec:RingSignature}, essentially following (often verbatimely)~\cite{DBLP:conf/eurocrypt/0001DHKS19}. 

\smallskip


\begin{defi}[Linkability] Let   $\mathit{LRS}=(\textsf{KeyGen, Sign, Verify, Link})$ be a \emph{linkable ring signature scheme}, and consider the following experiment ${\sf Exp_{LRS-Link}}$ involving a probabilistic polynomial time adversary $\mathcal{A}$:
\begin{enumerate}
\item $\mathcal{A}$ outputs $n$ tuples of the form $(pk_i,\sigma_i, m_i,R_i),$ for some $n$ polynomial in $\ell.$ At this, $R_1,\dots R_n$ and $R^*$ are sets of public keys (of finite size polynomial in $\ell$). Also, it outputs a triplet $(\sigma^*, m^*, R^*).$ Denote by $\mathit{PK}$ the set $\{pk_1,\dots, pk_n.\}$ 

\item The experiment outputs $1$ if the following conditions hold:
\begin{itemize}
    \item For all $i=1,\dots,n$ it holds $R_i\subseteq \mathit{PK}$ and also $R^*\subseteq \mathit{PK}.$
    \item For all $i=1,\dots, n$ both $\textsf{Verify}(R_i,m_i,\sigma_i)=1$ and $\textsf{Verify}(R^*,m^*,\sigma^*)=1$
   \item For all $i\ne j\in\{1,\dots,n\}$  we have $\textsf{Link}(\sigma_i,\sigma_j,m_i,m_j)=0$ and $\textsf{Link}(\sigma^*,\sigma_i,m^*,m_i)=0$
\end{itemize}
Otherwise, the experiment outputs $0.$
\end{enumerate}
We say $\mathit{LRS}$ attains \emph{linkability} if for every probabilistic polynomial time adversary $\mathcal{A},$ it holds that the above experiment outputs $1$ with negligible probability only.

\end{defi}

\smallskip 


\begin{defi}[Anonymity] \label{def:anon} Let   $\mathit{LRS}=(\textsf{KeyGen, Sign, Verify, Link})$ be a \emph{linkable ring signature scheme}, and consider the following experiment ${\sf Exp_{LRS-Anon}}$ involving a probabilistic polynomial time adversary $\mathcal{A}$ and a challenger $\mathcal{C}$:
\begin{enumerate}
\item On input the security parameter $\ell,$ $\mathcal{C}$ makes $n$ step by step executions of  $\textsf{Keygen}(1^{\ell}),$ for some $n\in \mathbb{N}$  polynomial in $\ell,$ and gets as output the corresponding pairs
$(pk_i,sk_i), i=1,\dots, n.$ Also, it stores the random coins involved in each key-pair generation $r_i.$ Further, it chooses a bit $b\in \{0,1\}$ uniformly at random. The public verification keys $\{pk_1,\dots, pk_n\}$ are forwarded to $\mathcal{A}.$ 
\item  $\mathcal{A}$ outputs a set of verification keys $U\subseteq \{pk_1,\dots, pk_n\}$ and two challenge verification keys $pk_0^*, pk_1^* \in \{pk_1,\dots, pk_n\}\setminus U.$ Both $U$ and $pk_0^*, pk_1^*$ are forwarded to the challenger $\mathcal{C}.$
\item $\mathcal{C}$ sends $r_i$, for each $pk_i\in U$ to the adversary (we thus assume the corresponding $sk_i$ are known to the adversary from this point on), 
\item Now, $\mathcal{A}$ sets a subset $R$ of public keys and interacts with the challenger, acting as a signing oracle. 
At this point, $\mathcal{A}$ queries (adaptively) for a polynomial number of signatures, each on a message $m$ of its choice and all on  input $R$ and a verification key $pk$ from $U,$ with the constraint that $pk \in  \{pk_1,\dots, pk_n\}\setminus U.$ The challenger $\mathcal{C}$ outputs:
\begin{itemize}
    \item a bit chosen uniformly at random, if $$|\{pk_0^*,pk_1^*\}\cap R| =1,$$ 
    \item $\textsf{Sign}(sk, m, R)$ -- where $sk$ is the signing key corresponding to $pk$ -- if $pk\notin \{pk_0^*,pk_1^*\},$
    \item $\textsf{Sign}(sk_b^*, m, R)$ if $pk = pk_0^*,$
    \item $\textsf{Sign}(sk_{1-b}^*, m, R)$ if $pk = pk_1^*.$
\end{itemize}
\item $\mathcal{A}$ outputs a bit $\beta$. Now the challenger sets as output for the experiment the bit  $1$ if $\beta =b$.
 \end{enumerate} 
We say $\mathit{LRS}$ attains \emph{anonymity} if for every probabilistic polynomial time adversary $\mathcal{A},$ it holds that the above experiment outputs $1$ with negligible advantage only, where the advantage is defined as 
$${\sf Adv_{LRS-Anon}}(\mathcal{A}) = |Pr[({\sf Exp_{LRS-Anon}}) =1] - {1}/{2}|.$$
\end{defi}

\smallskip 


\begin{defi}[Non-Frameability] Let   $\mathit{LRS}=(\textsf{KeyGen, Sign, Verify, Link})$ be a \emph{linkable ring signature scheme}, and consider the following experiment ${\sf Exp_{LRS-Frame}}$ involving a probabilistic polynomial time adversary $\mathcal{A}$ and a challenger $\mathcal{C}$:
\begin{itemize}
\item{Phase - I}
\begin{enumerate}
\item On input the security parameter $\ell,$ $\mathcal{C}$ makes $n$ step by step executions of  $\textsf{Keygen}(1^{\ell}),$ for some $n\in \mathbb{N}$  polynomial in $\ell,$ and gets as output the corresponding pairs 
$(pk_i,sk_i), i=1,\dots, n.$ Also, it stores the random coins involved in each key-pair generation $r_i.$ Further, it initializes a set $C = \emptyset$. 
The public verification keys $\{pk_1,\dots, pk_n\}$ are forwarded to $\mathcal{A}.$ 
\item $\mathcal{A}$ makes a polynomial number of queries to $\mathcal{C}$ :
\begin{itemize}
    \item Signing query $({\sf sign}, pk, m, R)$, for a message $m$ and a set of public keys $R.$  The challenger $\mathcal{C}$ will check if $pk\in R$ and, if so, output  $\textsf{Sign}(sk, m, R)$  where $sk$ is the matching secret key to $pk.$ 
    Otherwise, it outputs an error message $\perp.$ 
\item Corruption query $({\sf corrupt}, pk),$  for a public key $pk\in \{pk_1,\dots, pk_n\}$. The challenger  adds $pk$ to $C$ and returns to $\mathcal{A}$ the nonce  $r$ involved in its generation (for $\mathcal{A}$ may can be assumed to know the corresponding secret key $sk$). 
\end{itemize}
\item $\mathcal{A}$ outputs a triplet $(R^*,m^*,\sigma^*).$
\end{enumerate}
\item{Phase - II}
\begin{enumerate}
    \item The challenger provides all random coins $r_i, i=1,\dots, n$ to $\mathcal{A}$ (so now we may assume the adversary holds all involved secret keys),
    \item $\mathcal{A}$ outputs a triplet $(R^\dagger, m^\dagger,  \sigma^\dagger)$
     \item Now the challenger sets as output for the experiment the bit  $1$ provided that the following conditions hold:
     \begin{itemize}
         \item ${\sf Verify}(R^*,m^*, \sigma^*) = {\sf Verify} (R^\dagger, m^\dagger,  \sigma^\dagger)=1$ (i.e., both signatures output by $\mathcal{A}$ are valid), 
         \item $R^*\subseteq \{pk_1,\dots, pk_n\},$ and for all $i$ such that $pk_i\in R^*,$ it holds $pk_i\notin C$ (i.e., al verification keys in $R^*$ are from uncorrupted users),
         \item $\sigma^*$ was not part of the output of a signing query, 
         \item ${\sf Link}(\sigma^*,\sigma^\dagger, m^*, m^\dagger) =1.$
     \end{itemize}
    \end{enumerate}
\end{itemize}

Now, the scheme $\mathit{LRS}$ attains \emph{non-framability} provided that ${\sf Exp_{LRS-Frame}}(\ell)$ outputs $1$ with negligible probability only. 
\end{defi}


\smallskip 


\begin{defi}[Unforgeability]
Let   $\mathit{LRS}=(\textsf{KeyGen, Sign, Verify, Link})$ be a \emph{linkable ring signature scheme}, and consider the following experiment ${\sf Exp_{LRS-Unf}}$ involving a probabilistic polynomial time adversary $\mathcal{A}$ and a challenger $\mathcal{C}$:
\begin{enumerate}
\item On input the security parameter $\ell,$ $\mathcal{C}$ makes $n$ step by step executions of  $\textsf{Keygen}(1^{\ell}),$ for some $n\in \mathbb{N}$  polynomial in $\ell,$ and gets as output the corresponding pairs 
$(pk_i,sk_i), i=1,\dots, n.$ Also, it stores the random coins involved in each key-pair generation $r_i.$ Further, it initializes a set $C = \emptyset$. 
The public verification keys $\{pk_1,\dots, pk_n\}$ are forwarded to $\mathcal{A}.$ 
\item $\mathcal{A}$ makes a polynomial number of  signing queries to $\mathcal{C},$ described as follows: 
\begin{itemize}
    \item Signing query $({\sf sign},i, m, R)$, for a message $m,$ and index $i$ and a set of public keys $R.$  The challenger $\mathcal{C}$ will check if $pk_i\in R $ and, if so, output  $\sigma = \textsf{Sign}(sk_i, m, R)$  where $sk_i$ is the matching secret key to $pk_i.$
    Otherwise, it outputs an error message $\perp.$  
\item Corruption query $({\sf corrupt}, pk),$  for a public key $pk\in \{pk_1,\dots, pk_n\}$. The challenger  adds $pk$ to $C$ and returns to $\mathcal{A}$ the nonce  $r$ involved in its generation (we shall assumme thus that $sk$ is then accessible to $\cA$). 
\end{itemize}
\item $\mathcal{A}$ outputs a triplet $(R^*,m^*,\sigma^*).$ Now ${\sf Exp_{LRS-Unf}}$ outputs $1$ provided that:
\begin{itemize}
    \item $R^*\subseteq \{pk_1,\dots, pk_n\}$
    \item $\mathcal{A}$ never made a signing query of the form $({\sf sign}, \cdot, m^*, R^*)$
    \item ${\sf Verify}(R^*,m^*, \sigma^*) =1.$
\end{itemize}
\end{enumerate}
Now, the scheme $\mathit{LRS}$ attains \emph{unforgeability} provided that ${\sf Exp_{LRS-Frame}}(\ell)$ outputs $1$ with negligible probability only.

\end{defi}

\subsection{Security Analysis of our Linkable Ring Signature Scheme}
\label{appendix:security-proofs}


The linkable ring signature scheme we present in~\ref{sec:lsag-schema} can be proven to fulfill all aforementioned security definitions (under the random oracle model, and with the suitable assumptions on the underlying group). As we have noted, it is  essentially  a particular case of the scheme presented in~\cite{DBLP:conf/acisp/LiuWW04}, where the underlying cyclic group is taken from a ECC scenario, and can thus the subsequent discussion can be essentially derived  from the analysis in~\cite{DBLP:conf/acisp/LiuWW04, Liueprint}. There are however some differences in the related security definitions, for we have followed above the more recent terminology from~\cite{DBLP:conf/eurocrypt/0001DHKS19}. In particular, we here include non-framability in the presence of colluding malicious signers as a desirable property. Most importantly, following Backes et al., we allow for corruptions in the unforgeability definition, which is not considered in~\cite{DBLP:conf/acisp/LiuWW04}. This is particularly relevant for our use case; a collusion of malicious users of the system should indeed not be able to produce a valid signature that can be verified using a set of public keys from honest users. 

\smallskip 

\noindent{\bf Unforgeability.}
The following theorem proves the unforgeability of the proposed scheme.

\begin{theorem} The ring signature scheme presented in~\ref{sec:lsag-schema} attains  unforgeability in the random oracle model and under the DLOG assumption in the underlying cyclic group $E = (G)$. 
\end{theorem}

\begin{proof}
Let $\Adv(\cA)$ denote the advantage of the adversary against the unforgeability experiment, i.e.
$$ \Adv(\cA) = Pr[{\sf Exp_{LRS-Unf}} = 1] $$ 

We will structure the proof in a game-based fashion, namely, we start with a game that mimmicks the real attack  and end in a game in which the adversary's advantage is easily bounded,  and for which we can bound the
difference in the adversary's advantage between any two consecutive
games. Following standard notation, we denote by
$\Adv(\cA,G_i)$ the advantage of the adversary $\cA$ in Game~$i$.

\initexp This first game corresponds to a real attack.  By definition,
$\Adv(\cA,\gamecur) = \Adv(\cA).$

\newexp 
We now assume all involved hash functions behave in an ideal manner; namely, the outputs of $H,\mathit{H2P}, \mathit{HPK}$ are assumed to be chosen uniformly at random in the appropriate range for each new entry, which is stored in a corresponding hash list so that subsequent queries are answered consistently. Queries to these hash functions will be modelled as queries to random oracles $\cO_H, \cO=_{\mathit{H2P}}, \cO_{\mathit{HPK}}.$
We make explicit the so-called \emph{random oracle assumption} by assuming 
$\Adv(\cA,\gamecur) = \Adv(\cA, \gameold).$

\newexp 
Let us assume the 
$\Adv(\cA,\gamecur) \ge \frac{1}{p(\ell)},$ for some polynomial $p.$ Namely, $\cA$, on input a certain set of ring public keys  $\{pk_1,\dots, pk_n\}$ it outputs a valid signature $(R^*,m^*,\sigma^*)$ with non-negligible probability. We will see how to construct a machine $\cS,$  a simulator using $\cA$ as a subroutine which, on input a DLOG instance on the group $E$ $(p,g, G, Y)$ will output in polynomial time a solution, i.e.,  $x\in\{1,\dots, g-1\}$   $Y = x\cdot G.$

At this, $\cS$ must simulate the challenger from the ${\sf Exp_{LRS-Unf}}$ game for $\cA.$ Lets see how to do this.

Following the strategy from Theorem 4.2. in~\cite{TesisJavi}, we define $\mu = (\frac{5}{6})^{\mathcal{Q}_C},$ where 
$\mathcal{Q}_C$ is a bound on the number of corrupt calls from $\cA$ (and is indeed polynomial in the security parameter).  SHould one restrict to adversaries which may not corrupt any user, then set $\mu =0.$ 

Our simulator $\cS$ will creat a table CORR in the following way: for each $i=1,\dots, n$ a key pair is set by choosing a bit $c_i$ (as $c_i=0$ with probability $\mu$ and $c_i=1$ with probability $1-\mu$):  
\begin{itemize}
    \item if $c_i=0,$ $\cS$ selects uniformly at random  $sk_i\in\{1,\dots, g-1\}$ and sets   $pk_i= sk_i\cdot G,$
    \item if $c_i=1,$ $\cS$ selects uniformly at random  $alpha_i\in\{1,\dots, g-1\}$ and sets   $pk_i= alpha_i\cdot Y,$
\end{itemize}

The entries $(i, pk_i, sk_i, c_i)$ are stored in the CORR table. Now the attacker is presented with the public keys $pk_1,\dots, pk_n$ for running $\sf Exp_{LRS-Unf}.$

Now whenever $\cA$ makes a query of the form $({\sf Corrupt}, pk_i)$, for $i=1,\dots, n$,
 $\cS$ looks up in the CORR table and, if $c_i=1,$ forwards the answer $\alpha_i$ to $\cA,$ halting however  if $c_i=1.$ Note that the probability that $\cS$ halts is thus $1-\mu^{\mathcal{Q}_C}\le \frac{1}{6}.$

Now, assume the simulator $\cS$ is called with a query $(\textsf{sign}, i, m, \{pk_1',\dots, pk_m'\})$ where each $pk_i'\in \{pk_1,\dots, pk_n\}$ and the attacker has not queried $(\textsf{Corrupt}, pk_i').$

Assuming the simulator is actually controlling the random oracle $\mathit{H2H}$, it is clear that it  may simply be simulated by choosing,  on a query $X,$  a random values $r_X\in \{1,\dots, g-1\}$ and outputting the curve point $r_x\cdot G.$ As a result, we may assume the simulator  actually holds $r \in \{1,\dots, g-1\}$ so that $$L=\mathit{H2H}(\mathit{HPK}((pk_1,\dots, pk_n))=  r \cdot G.$$ Note then that for each $i= 1,\dots, m,$ it holds $r\cdot pk_i' = sk_i'\cdot L,$ from which $\cS$ can compute the needed values for signing in spite of  ignoring the corresponding secret keys (for by construction, $sk_i' =\alpha_i'x,$ for some $x$ such that $Y = x\cdot G$). 

Now, here is how $\cS$ constructs an answer for the given query:
\begin{enumerate}
          \item $\cS$ chooses at random an index $\pi\in \{1,\dots, m\},$
        \item For each $i\in \{1,\dots m\},$chooses independently and uniformly at random $s_i \in \{1,\dots, g-1\}$  
     and  $c_1,\dots, c_m \in \mathbb{F}_p$
                       \item back-patch the random oracle so that for each $i=1,\dots, m$ it holds:
                       $$A_i =s_i\cdot G + c_{i-1}\cdot pk_i', B_i = s_i\cdot r \cdot  pk_\pi' + c_{i-1}\cdot  r\cdot pk_\pi' . $$ 
                       \item Set $T= r\cdot pk_\pi'$, output $$\sigma =(pk_1',\dots, pk_m', T, s_1,\dots, s_m, c)$$
                       \end{enumerate}

Note that now, as argued in~\cite{Liueprint} for the signature simulations and considering the halting probability of $\cS,$ not all invocations of $\cA$ will produce a successful forgery, yet it will indeed do so  with non negligible probability. 

Furthermore, either mimicking the rationale from~\cite{Liueprint} or relying on the so-called \emph{forking lemma for ring signatures} as introduced in~\cite{TesisJavi}, we conclude that as $\cA$ is able to produce a successfull forgery, he can be used by $\cS$ to get two different forgeries w.r.t. the same ring of users $R=\{pk_1',\dots, pk_m'\},$ on the same message, reusing the same random values inputs. 

Let the corresponding signatures be:
$$\sigma = (R, T, s_1,\dots, s_m, c),$$\mbox{ with associated hashes } $c_1,\dots, c_n$
and 

$$\sigma^* = (R, T^*, s_1^*,\dots, s_m^*, c^*)$$
 \mbox{ with associated hashes } $c_1^*,\dots, c_m^*.$

From the aforementioned \emph{forking lemma} it yields  that there must be a value $\pi \in \{1,\dots, m\}$ such that the corresponding $c_\pi = c_\pi^*$, and, 
as a result it must be 
$A_\pi = A_\pi^*$ and $B_\pi=B_\pi^*$ thus, in particular, 
$$
A_{\pi} = s_{\pi} \cdot G + c_{\pi-1} \cdot pk_{\pi}' = 
s^*_{\pi} \cdot G + c_{\pi-1}^*\cdot pk_{\pi}' = A_{\pi}^*.$$

Now as $pk_{\pi}' = \alpha_{\pi}' \cdot Y $ we have 

$$x = \frac{s_\pi -s_\pi^*}{\alpha_{\pi}'(c_{\pi-1}^* - c_{\pi -1})} \pmod q,$$ which fulfills $x \cdot G = Y.$  Summing up, $\cS$ is able to output thus the solution $x$ as above with non negligible probability, so the desired result is proven.
\end{proof}

\smallskip 

\noindent{\bf Linkability.}
Informally, it is easy to see why the scheme from Section ~\ref{sec:lsag-schema} is unlinkable; indeed, to construct  two correct signatures linking to different users is to actually have their private keys, so the hardness of DLOG should prevent an adversary from breaking linkability. More formally, Liu et al. state in Theorem 3 of~\cite{Liueprint} an analagous of the  result we prove below (essentially adapting their rationale):   

\begin{theorem}  The ring signature scheme presented in~\ref{sec:lsag-schema} attains linkability in the random oracle model and under the DLOG assumption in the underlying cyclic group $E = (G)$. 
\end{theorem}

\begin{proof}
Let $\Adv(\cA)$ denote the advantage of the adversary against the linkability experiment, i.e.
$$ \Adv(\cA) = Pr[{\sf Exp_{LRS-Link}} = 1] $$ 

Once again, we will structure the proof in a game-based fashion and denote by
$\Adv(\cA,G_i)$ the advantage of the adversary $\cA$ in Game~$i$.

\initexp This first game corresponds to a real attack.  By definition,
$\Adv(\cA,\gamecur) = \Adv(\cA).$

\newexp 
Again, we assume the involved hash functions behave in an ideal manner; their queries will thus be modelled as queries to random oracles $\cO_H, \cO=_{\mathit{H2P}}, \cO_{\mathit{HPK}}.$ The  \emph{random oracle assumption} is thus made explicit by assuming 
$\Adv(\cA,\gamecur) = \Adv(\cA, \gameold).$

Now we can argue that  $\Adv(\cA,\gamecur)$ is bounded by the probability of success of $\cA$ in solving DLOG in $E,$ just as it is done in the proof of Theorem 3 from~\cite{DBLP:conf/acisp/LiuWW04}. Indeed, let us define a simulator $\cS$ which will compute the DLOG of one of the public keys $\{pk_1,\dots, pk_n\}$ involved in the linkability experiment ${\sf Exp_{LRS-Link}}$. We will again make use of a so-called forking argument. Indeed, our simulator $\cS$ interacts with $\cA$ in the linkability experiment, and so $\cA$ succeeds by outputting a pair of signatures $\sigma_i,\sigma_j$ which are not linked and for which $\cA$ only was initially given one corresponding secret key (either $sk_i$ or $sk_j$). Let $\sigma$ and $\sigma'$ be the mentioned signatures, and $T$ and $T'$ be the corresponding linkability tags (namely, of the form $sk_i \cdot L$ and $sk_k \cdot L'$ (note that the signatures may not be verifiable with the same set of public keys, yet both the ring of users corresponding to $\sigma$ and the ring of users corresponding to $\sigma'$ have public keys contained in $pk_1,\dots, pk_n.$ It is easy to see that, with non-negligible probability, $\sigma$ and $\sigma'$ may correspond to two different public keys $\pi$ and $\pi'$ (see the argument contained in the proof of Theorem 3 from~\cite{DBLP:conf/acisp/LiuWW04}).

As a result, using the forking lemma twice (once for $\sigma,$ second time for $\sigma'$) it follows that $\cA$ can, by the forking lemma for ring signatures (see~\cite{Javi, TesisJavi}), produce four different forgeries: 
- two for the same set of public keys related to $\sigma$
- two for the same set of public keys related to $\sigma'.$

As a result, by the same argument used in the proof of unforgeability above, $\cA$ can both retrieve $sk_\pi$ and $sk_{pi'};$ given that, by assumption, only one of these secret keys was already known to $\cA,$ to the adversary has indeed been able to solve DLOG for one of the public keys in $\{pk_1,\dots, pk_n\}.$
\end{proof}

\smallskip 

\noindent{\bf Anonymity.}
%
The following result can be attained following  Liu et al.s rationale from  Theorem 2. As we are using a different definition we include a sketch the proof below, which is a straightforward adaptation of the one that can be found in~\cite{Liueprint}. 

\begin{theorem}  The ring signature scheme presented in~\ref{sec:lsag-schema} attains anonymity in the random oracle model and under the DDH assumption in the underlying cyclic group $E = (G)$. 
\end{theorem}

\begin{proof} (sketch) 
Let $\Adv(\cA)$ denote the advantage of the adversary against the anonymity experiment, i.e.
$$ \Adv(\cA) = Pr[{\sf Exp_{LRS-Anon}} = 1] $$ 

Once again we structure the proof in a game-based fashion, and we denote by
$\Adv(\cA,G_i)$ the advantage of the adversary $\cA$ in Game~$i$.

\initexp This first game corresponds to a real attack.  By definition,
$\Adv(\cA,\gamecur) = \Adv(\cA).$

\newexp 
All involved hash functions are assumed to behave in an ideal manner;as before, by the random oracle assumption, that 
$\Adv(\cA,\gamecur) = \Adv(\cA, \gameold).$ Now we can argue that  $\Adv(\cA,\gamecur)$ is bounded by the probability of success of $\cA$ in solving DDH in $E.$ The main idea behind this argumentation is that a polynomial time distinguisher $\cD$ for the DDH in $E$ can be construced from a succesful adversary $\cA$ winning the anonymity game from Definition~\ref{def:anon}. 

Let $\cD$ be a probabilistic polynomial time algorithm that gets as input a triplet $(\alpha, \beta, \gamma)\in E^3,$ so that there exist uniformly at random values $a,b,c\in \{1,\dots, g-1\}$ such that either $\alpha = a\cdot G,$ $\beta = b\cdot G $ and $\gamma = c \cdot G$ or
$\alpha = a\cdot G,$ $\beta = b\cdot G $ and $\gamma = ab \cdot G$. Our goal is to prove that  $\cD,$ using a successful linkability adversary $\cA,$  is able to distinguish (with probability non-negligibly above $\frac{1}{2}$) which is the case. Let then $\cD$  act as an anonymity challenger for $\cA$ as follows: 

\begin{enumerate}
       \item $\cD$ initially sets   $pk_\pi = \alpha,$ and selects $pk_i$ for $i\ne \pi$ choosing uniformly at random secret keys $sk_i \in \{1,\dots, g-1\}$,
 \item for $i = \pi,...,n,1,\dots, \pi-1,$ $\cD$  selects random $s_i\in \{1,\dots, g-1\}$ and sets 
  $$c_{i+1}= H(m || A_i || B_i)$$
 with $A_i,B_i$ computed as in the protocol specification form the above known values but fixing (as $\cD$ ignores the secret key  $sk_\pi$) the values $T= \gamma$ and   back-patching on the random oracle on $L$, i.e., fixing $L = \beta.$

\item $\cD$ outputs a signature 
$$\sigma = (pk_1,\dots, pk_n,T, s_1,\dots, s_n,c_n).$$
  \end{enumerate}

Now in the simulation $\cD$ calls $\cA$ with the corresponding challenge and simulates the random oracles for him faithfully (except for the query $L$ and queries involved in the above construction of the $c_i's,$ where consistency is mantained). 

Now, suppose $\cA$ outputs, after interacting with $\cD,$ a bit $\hat{b}$. Indeed, $\cD$ will output $1$ if indeed $\hat{b}$ is a pointer to $\pi.$ Clearly, if the input triplet $(\alpha, \beta, \gamma)$ is not a Diffie-Hellman triplet, all signers will be identical for $\cA$ and  thus its advantage is exactly $\frac{1}{2}.$ However, if the given triplet is indeed a Diffie-Hellman triplet, $\cA$ will have a non-negligible advantage, and so will, by a simple probability argument, have $\cD$ in solving the corresponding DDH instance. 
\end{proof}

\smallskip 

\noindent{\bf Non-Framability.}
Intuitively, it is easy to see that an adversary violating non-framability must indeed have succeeded in a forging challenge, so non-framability is obtained from unforgeability.  

\begin{theorem}  The ring signature scheme presented in~\ref{sec:lsag-schema} achieves non-framability in the random oracle model and under the DLOG assumption in the underlying cyclic group $E = (G)$. 
\end{theorem}

\begin{proof}
Out main argument is that for ${\sf Exp_{LRS-Frame}}$ to output $1,$ $\cA$ must have actually forged a signature under one of the public verification keys presented by the challenger. 
Indeed, notice that in Step 3 of  Phase-1 of ${\sf Exp_{LRS-Frame}}$  the adversary $\cA$ must output a triplet $(R^*, m^*, \sigma^*)$ from identical inputs to those presented in an execution of ${\sf Exp_{LRS-Unf}},$ i.e. : 
\begin{enumerate}
\item On input the security parameter $\ell,$ the challenger $\mathcal{C}$ makes $n$ step by step executions of  $\textsf{Keygen}(1^{\ell}),$ for some $n\in \mathbb{N}$  polynomial in $\ell,$ and gets as output the corresponding pairs $(pk_i,sk_i), i=1,\dots, n.$ Also, it stores the random coins involved in each key-pair generation $r_i.$ Further, it initializes a set $C = \emptyset$.
The public verification keys $\{pk_1,\dots, pk_n\}$ are forwarded to $\mathcal{A}.$ 
\item $\mathcal{A}$ makes a polynomial number of queries to $\mathcal{C}$:
\begin{itemize}
    \item Signing query $({\sf sign}, pk, m, R)$, for a message $m$ and a set of public keys $R.$  The challenger $\mathcal{C}$ will check if $pk\in\{pk_1,\dots, pk_n\} $ and, if so, output  $\textsf{Sign}(sk, m, R)$  where $sk$ is the matching secret key to $pk.$ 
    Otherwise, it outputs an error message $\perp.$ 
\item Corruption query $({\sf corrupt}, pk),$  for a public key $pk\in \{pk_1,\dots, pk_n\}$. Now $\mathcal{C}$  adds $pk$ to $C$ and returns to $\mathcal{A}$ the corresponding secret key $sk$. 
\end{itemize}
\item $\mathcal{A}$ outputs a triplet $(R^*,m^*,\sigma^*).$
\end{enumerate}

Indeed, now the triplet  $(R^*,m^*,\sigma^*)$  is a winning triplet in ${\sf Exp_{LRS-Unf}},$ for 
     \begin{itemize}
         \item ${\sf Verify}(R^*,m^*, \sigma^*) = 1$  
         \item $R^*\subseteq \{pk_1,\dots, pk_n\}\setminus C$ (i.e., al verification keys in $R^*$ are from uncorrupted users),
         \item $\sigma^*$ was not part of the output of a signing query. 
              \end{itemize}
Thus, non-framability is derived directly from unforgeability, and so if follows from  the DLOG assumption in the underlying cyclic group $E = (G)$. 
\end{proof}

\end{document}